\documentclass[letterpaper, 10 pt, conference]{ieeeconf}  

\IEEEoverridecommandlockouts                              
\usepackage{graphics} 
\usepackage{epsfig} 
\usepackage{mathptmx} 
\usepackage{times} 

\usepackage{amsthm}
\usepackage{amsmath}
\usepackage{amssymb}  
\usepackage{amsfonts}
\usepackage{algpseudocode}
\usepackage{algorithm}
\usepackage{xcolor}
\usepackage{balance}
\usepackage{microtype}

\usepackage{cite}

\usepackage{graphicx}
\usepackage{float,epsf,caption,subcaption}

\newtheorem{theorem}{Theorem}
\newtheorem{assumption}{Assumption}

\theoremstyle{plain}
\newtheorem{lemma}{Lemma}
\theoremstyle{definition}
\newtheorem{definition}{Definition}
\newtheorem{example}{Example}
\theoremstyle{remark}

\newcommand{\R}{\mathbb{R}}

\makeatletter
\def\endthebibliography{%
	\def\@noitemerr{\@latex@warning{Empty 'thebibliography' environment}}%
	\endlist
}
\makeatother
\title{\LARGE \bf
On the Convergence of the Backward Reachable Sets of Robust Controlled Invariant Sets For Discrete-time Linear Systems
}

\author{Zexiang Liu, Necmiye Ozay
	\thanks{Zexiang Liu and Necmiye Ozay are with the Dept. of Electrical Engineering and Computer Science, Univ. of Michigan, Ann Arbor,	MI 48109, USA 
		{\tt\small zexiang,necmiye@umich.edu}. This work was supported by NSF Grant CCF-1918123.
	}
}

\begin{document}
	\maketitle
	\thispagestyle{empty}
	\pagestyle{empty}

	\begin{abstract}
		This paper considers discrete-time linear systems with bounded additive disturbances, and studies the convergence properties of the backward reachable sets of robust controlled invariant sets (RCIS). Under a simple condition, we prove that the backward reachable sets of an RCIS are guaranteed to converge to the maximal RCIS in Hausdorff distance, with an exponential convergence rate. When all sets are represented by polytopes, this condition can be checked numerically via a linear program. We discuss how the developed condition generalizes the existing conditions in the literature for (controlled) invariant sets of systems without disturbances (or without control inputs).
	\end{abstract}
	
\section{Introduction}	
Robust controlled invariant sets (RCISs) have a wide range of applications in control and verification for safety-critical systems, such as autonomous vehicles\cite{kianfar2013safety,falcone2011predictive,nilsson2015correct,chou2018using}, robotic systems\cite{singletary2019online,ames2019control} and power systems\cite{jang2021large}. 
In many lines of works regarding RCIS, the backward reachable sets of RCISs need to be computed implicitly or explicitly. For instance, the computation of the maximal RCIS is known to be a difficult problem \cite{bertsekas1972infinite, de2004computation, rungger2017computing}. A common means to efficiently inner-approximate the maximal RCIS is by computing the backward reachable sets of a small but easy-to-compute RCIS\cite{de2004computation, rungger2017computing,fiacchini2017computing,wintenberg2020implicit,anevlavis2021enhanced}, which is referred to as the \emph{inside-out} algorithm in this work.
 In the context of the constrained model predictive control (MPC), the terminal set of the MPC is typically chosen as an RCIS to guarantee constraint satisfaction and recursive feasibility\cite{blanchini2008set,limon2005enlarging,jalalmaab2017guaranteeing}. Then, the domain of attraction (DoA) of the MPC is exactly equal to the $T$-step backward reachable set of the terminal set, where $T$ is the prediction horizon\cite{limon2005enlarging, rakovic2009set}. More recently, inspired by MPC, a control framework called safety filter is proposed for discrete-time systems \cite{wabersich2018linear} and continuous-time systems \cite{singletary2019online,gurriet2020scalable}  to equip a given nominal controller with safety guarantees in a minimally invasive way. In discrete-time setup, the connection between the DoA of a safety filter and the backward reachable sets of RCISs are similar to that in MPC. In continuous-time setup, a finite-step backward reachable set of some given RCIS is constructed implicitly, to improve the scalability and reduce the conservatism of the safety filter\cite{singletary2019online,gurriet2020scalable}.

Due to the role of the backward reachable sets of RCISs in all the aforementioned works, it is important to understand the convergence properties of  the backward reachable sets of RCISs, including (i) what conditions ensure that the backward reachable sets of an RCIS converges to the maximal RCIS and (ii) how fast the backward reachable sets of an RCIS converges.  In the literature, the existing works in this direction can only deal with systems without disturbances: Ahmadi and Gunluk \cite{ahmadi2018robust} study asymptotically stable autonomous systems and provide a sufficient condition under which the $k$-step backward reachable set of an invariant set converges to the maximal invariant set in finite steps. The papers \cite{cwikel1986convergence, gutman1987algorithm, de2004computation,darup2014general} study controllable systems without disturbances and show a sufficient condition under which the $k$-step backward reachable set of a controlled invariant set (CIS) converges to the maximal CIS in Hausdorff distance, with an exponential convergence rate. When the system is asymptotically stabilizable without disturbances, Santis et al. \cite{de2004computation} present a sufficient condition under which the backward reachable sets of a small CIS converges to the maximal CIS, but does not reveal the convergence rate. 

As the main contribution of this work, for linear systems with both control and disturbance inputs, we develop a mild sufficient condition under which the $k$-step backward reachable set of an RCIS converges to the maximal RCIS in Hausdorff distance, with an exponential convergence rate.   Furthermore, when all sets are represented by polytopes, this sufficient condition can be easily checked by a linear program. To the best of our knowledge, our work is the first result in the literature that shows these convergence properties for systems with additive disturbances. 
When restricted to systems without disturbances, the existing results in \cite{gutman1987algorithm, de2004computation,darup2014general, cwikel1986convergence} are shown to be special cases of our result. In addition, our result extends the results for asymptotically stabilizable systems in \cite{de2004computation} by showing the convergence rate is exponential.

In the remainder of this work, the preliminaries of RCISs are introduced in Section \ref{sec:prelim}. Then, our main result is stated in Section \ref{sec:main} and proven in Section \ref{sec:pf}. After that, we compare our result with existing results in the literature in Section \ref{sec:dis} and conclude the paper in Section \ref{sec:con}. The proofs of all the minor theorems are found in Appendix. 

\noindent \textbf{Notation:} The multiplication $AX$ of a matrix $A$ and a set $X$ is defined by ${AX = \{Ax \mid x\in X\}}$. The projection of a set  $X \subseteq \R^{n+m}$ onto the first $n$ coordinates is denoted by $\pi_{[1,n]}(X)=\{x_1\in \R^{n} \mid (x_1,x_2)\in X\} $. We define the distance from a point $x$ to a set $Y$ by $d(x,Y) = \inf\{ \Vert x-y\Vert_{2} \mid y\in Y\}$, and then define the Hausdorff distance of two sets $X$ and $Y$ by $d(X,Y) = \max(\sup_{x\in X}d(x,Y), \sup_{y\in Y}d(y,X))$. The Minkowski sum of two sets $X$ and $Y$ is denoted by  $X+Y=\{x+y \mid x\in X, y\in Y\} $. When $Y=\{y\} $ is a singleton, the sum $X+Y$ is written as $X+y$ for short. Also, for a set $X$ and a point $y$, $X-y = X+(-y)$. Given a scalar $ \alpha$ and a set $X$, the set $\alpha X = \{ ax \mid x\in X\} $.  For a subset $X$ of  $\R^{n}$, the interior and closure of $X$ are denoted by $int(X)$ and $cl(X)$. The $ \epsilon$-ball centered at $x$ is denoted by $B_{ \epsilon}(x) = \{y \mid \Vert x-y\Vert_2 < \epsilon\} $. 
\section{Preliminaries} \label{sec:prelim} 
We consider a discrete-time linear system $\Sigma$ in form of 
\begin{align} \label{eqn:sys} 
\Sigma:   x(t+1) = A x(t) + Bu(t) + Ed(t),
\end{align}
with $x\in \R^{n}$, $u\in \R^{m}$, and $d\in D \subseteq\R^{l}$. The state-input constraints of the system $\Sigma$ are specified by the \emph{safe set} $S_{xu}$, which contains all the desired state-input pairs $(x,u)$.  

\begin{assumption} \label{asp:cpt} 
The safe set $S_{xu}$ and the disturbance set $D$ are convex and compact.
\end{assumption}

\begin{definition} \label{def:rcis} 
    A set $C \subseteq \R^{n}$ is a \emph{robust controlled invariant set} (RCIS) of the system $\Sigma$ in the safe set $S_{xu}$ if for all  $x\in C$, there exists $u\in \R^m$ such that  ${(x,u)\in S_{xu}}$ and ${Ax+Bu+ ED \subseteq C}$. 

	An RCIS $C_{max}$ is the maximal RCIS of $\Sigma$ in $S_{xu}$ if $C_{max}$ contains any RCIS of $\Sigma$ in $S_{xu}$. 
\end{definition}
If the disturbance set $D=\{0\}$, a set $C$ satisfying conditions in Definition \ref{def:rcis} is just called \emph{controlled invariant set} (CIS) by convention.

\begin{definition} \label{def:pre} 
Given a set $X \subseteq \R^{n}$, the \emph{one-step backward reachable set} $Pre_{\Sigma}(X,S_{xu},D)$ of $X$ with respect to the system $\Sigma$, the safe set  $S_{xu}$ and the disturbance set $D$ is the set of states $x$ where there exists an input $u$ such that the state-input pair $(x,u)$ is in the safe set $S_{xu}$ and the next state stays in $X$ robust to any disturbance in $D$, that is
\begin{align} \label{eqn:pre} 
    Pre_{\Sigma}(X,S_{xu},D) = \{x \mid \exists u, (x,u)\in S_{xu}, Ax+Bu+ED \subseteq X\}.
\end{align}
\end{definition}
By \cite[Proposition 1]{rakovic2006reachability}, a set $C$ is an RCIS of $\Sigma$ in $S_{xu}$ if and only if
\begin{align} \label{eqn:rcis_pre} 
    C \subseteq Pre_{\Sigma}(C,S_{xu},D).
\end{align}
Furthermore, if $C$ is an RCIS of $\Sigma$ in $S_{xu}$, then $Pre_{\Sigma}(C, S_{xu},D)$ is also an RCIS of $\Sigma$ in $S_{xu}$. Thus, a set  $C_{max}$ is the maximal RCIS only if
\begin{align} \label{eqn:mrcis_pre} 
    C_{max} = Pre_{\Sigma}(C_{max}, S_{xu},D).
\end{align}
Given Definition \ref{def:pre}, we define the \emph{$k$-step backward reachable set} $Pre^{k}_{\Sigma}(C_0,S_{xu},D)$ of a target set $C_0$ recursively by
\begin{align}
	\begin{split} \label{eqn:k_pre} 
   Pre_{\Sigma}^{1}(C_0,S_{xu},D) &=  Pre_{\Sigma}(C_0, S_{xu}, D),\\
Pre_{\Sigma}^{k}(C_0,S_{xu},D) &= Pre_{\Sigma}(Pre_{\Sigma}^{k-1}(C_0,S_{xu},D), S_{xu}, D).
	\end{split}
\end{align}
When $C_0$, $S_{xu}$ and $D$ are clear in context, we denote $Pre_{\Sigma}^{k}(C_0,S_{xu},D)$ by $C_{k}$ for short. 
Based on different choices of $C_0$, there are two standard algorithms that compute approximations of the maximal RCIS: 

\noindent \textbf{The outside-in algorithm}\cite{bertsekas1972infinite, rungger2017computing}: Let $C_{0} = \pi_{[1,n]}(S_{xu})$ be the projection of the safe set $S_{xu}$ onto the $x$ coordinates. Recursively compute $C_{k}$ in \eqref{eqn:k_pre} until the termination condition $C_{k}=C_{k-1}$ is satisfied.

\noindent \textbf{The inside-out algorithm}\cite{de2004computation,rungger2017computing}: Let $C_{0}$ be a known RCIS of $\Sigma$ in $S_{xu}$. Recursively compute $C_{k}$ in \eqref{eqn:k_pre} until the termination condition $C_{k}=C_{k-1}$ is satisfied or  $k$ reaches a user-defined maximum step $k_{max}$. 

The $k$-step backward reachable set $C_{k}$ in the outside-in algorithm outer-approximates the maximal RCIS\cite{bertsekas1972infinite}, while the set $C_{k}$ in the inside-out algorithm inner-approximates the maximal RCIS. 
Both algorithms may fail to terminate within finite steps\cite{rungger2017computing}. However, when prematurely terminated at step $k$, the set $C_{k}$ in the outside-in algorithm is not necessarily robust controlled invariant, and thus is less useful in control synthesis; in contrast, the set $C_{k}$ in the inside-out algorithm is always an RCIS. Due to this reason, the inside-out algorithm is also called an \emph{anytime} algorithm \cite{rungger2017computing}, as users can stop the algorithm at any step $k$ and use the RCIS $C_{k}$ in control synthesis. 

The inside-out algorithm is computationally more appealing, but the convergence properties of the inside-out algorithm are poorly studied, which motivates our work. 
Thus, for the remainder of this work, if not stated otherwise, we assume $C_0$ is an RCIS of $\Sigma$ within the safe set $S_{xu}$, which is the setup for the inside-out algorithm.

Given $C_0$ is an RCIS, the set $\{C_{k}\}_{k=0}^{ \infty}$ of finite-step backward reachable sets of $C_0$ is an expanding family of RCISs. That is,  for all $k\geq 0$, $C_{k}$ is robust controlled invariant in $S_{xu}$, and $C_{k+1} \supseteq C_{k}$. Thus, the limit $C_{ \infty}$ of the $k$-step backward reachable set of $C_0$ as $k$ goes infinity is well defined:
\begin{align} \label{eqn:C_inf} 
    C_{ \infty} = \cup_{k=0}^{ \infty} C_{k}.
\end{align}
In this work, we also refer to the limit $C_{ \infty}$ as the \emph{infinite-step} backward reachable set of $C_0$. 
It can be shown that  $C_{ \infty}$ is an RCIS within $S_{xu}$.  

Formally, in this work, we want to answer the following two questions: 

\noindent\fbox{%
    \parbox{0.47\textwidth}{%
    \emph{(i) What condition ensures that the $k$-step backward reachable set $ C_{k}$ converges to the maximal RCIS $C_{max}$ in Hausdorff distance, that is, $d(C_{ \infty}, C_{max}) = 0$? 
    (ii) How fast does the $k$-step backward reachable set $C_{k}$ converges to its limit $C_{ \infty}$?}
}}

\section{Main Result} \label{sec:main} 
Before we state our main result, let us first gain some intuitions from a toy example.

\begin{example} \label{ex:1} 
Consider the $1$-dimensional system
\begin{align} \label{eqn:sys_1d} 
    x(t+1)= \alpha x(t) + u(t)+d(t),
\end{align}
with $x\in \R$, $u\in \R$, and $d\in[-d_{max},d_{max}]$. The safe set is ${S_{xu}=[-x_{max},x_{max}]\times [-u_{max},u_{max}]}$. Define 
\begin{align}
c_d = (d_{max}-u_{max})/(1- \vert \alpha \vert ).
\end{align}
Consider symmetric RCIS $C_0 $ in form of $ [-c_0,c_0]$ with $c_0 \leq x_{max}$. Then, the $k$-step backward reachable set $C_{k}$ of $C_0$ is symmetric, and if $ \vert \alpha \vert\not\in \{0,1\}$, $C_k$ is equal to $[-c_{k},c_{k}]$ with
\begin{align} \label{eqn:ck_1d} 
   c_{k} =  \min\left(\frac{c_0-c_d}{ \vert \alpha \vert^{k}}+ c_d,\ x_{max}\right).
\end{align}

\underline{Case $1$}: Suppose $ \vert\alpha \vert\in (0,1)$, $u_{max}\leq d_{max}$ and ${x_{max} > c_d\geq d_{max}}$. For any $c_0\in [ c_d, x_{max}]$, $C_0= [-c_0,c_0]$ is an RCIS. The maximal RCIS is $[-x_{max},x_{max}]$. 

According to \eqref{eqn:ck_1d}, if we select $c_0\in (c_d, x_{max})$, there always exists a finite $k$ such that $c_{k}$ is equal to $x_{max}$. That is, the $k$-step backward reachable set $C_{k}$ converges to the maximal RCIS $[-x_{max}, x_{max}]$ in finite steps. 
However,  if we select $c_0=c_d$, then $c_{k}= c_{0} < x_{max}$ for all $k\geq 0$. That is, $C_{k}=C_{0}$ fails to converge to the maximal RCIS. 

\underline{Case 2}: Suppose that $ \vert \alpha \vert >1$, $u_{max}\geq d_{max}$, and ${x_{max} \geq  c_d\geq d_{max}}$. For any $c_0\in [d_{max}, c_d]$, $C_0 = [-c_0,c_0]$ is an RCIS. The maximal RCIS is $[-c_d, c_d]$. 

According to \eqref{eqn:ck_1d}, if we select any $c_0\in [d_{max}, c_d)$, $c_{k}$ converges to $c_d$ in the limit. That is, the $k$-step backward reachable set $C_{k}$ converges to the maximal RCIS in Hausdorff distance in infinite steps. The limit $C_{ \infty}=(-c_d,c_d)$ is the interior of the maximal RCIS. Furthermore, the Hausdorff distance between  $C_{k}$ and the maximal RCIS is
\begin{align} \label{eqn:conv_rate_1d} 
    d(C_{k}, [-c_d,c_d]) = \frac{c_d-c_0}{ \vert \alpha \vert^{k}}, 
\end{align}
which decays to $0$ exponentially fast.

\underline{Otherwise:} For all the other cases where the maximal RCIS is not empty, the behavior of $C_k$ is similar to Case 1. That is, $C_k$ is either equal to $C_0$ for all $k\geq 0$ or converges to the maximal RCIS in finite steps. 
\end{example}

In Example \ref{ex:1}, we observe three types of limit of $C_{k}$: (i) the limit $C_{ \infty}$ is exactly equal to the maximal RCIS $C_{max}$; (ii) $C_{ \infty}$ is a subset of $C_{max}$, but the Hausdorff distance $d(C_{ \infty}, C_{max}) = 0$; (iii) $C_{ \infty}$ is a subset of $C_{max}$ and the Hausdorff distance $d(C_{ \infty},C_{max}) >0$. Note that for the limit type (ii), even if $C_{ \infty} \subset C_{max}$, the maximal RCIS  $C_{max}$ is equal to the closure of $C_{ \infty}$  since $d(C_{ \infty}, C_{max}) = 0$. Thus, among the three types, the limit type (iii) is the least desirable one, as in this case it is impossible to obtain the maximal RCIS from the backward reachable sets of $C_0$. 

A key observation in Example \ref{ex:1} that distinguishes between the limit types (i), (ii) and the limit type (iii) is that if there exists a $k$ such that $C_{k}$ contains $C_0$ in the interior, then $C_{ \infty}$ is either in type (i) or in type (ii).  Indeed, in Example \ref{ex:1}, the limit type (iii) only happens when $C_{k} = C_0$ for all $k\geq 0$. Intuitively, this observation suggests that if the backward reachable sets of $C_0$ expand in all directions in $\R^{n}$, then they keep expanding until they reach the maximal RCIS (that is the limit types (i) and (ii)). In other words, the limit type (iii) occurs only if the backward reachable sets of $C_0$ only expand in certain directions, or do not expand at all (which happens in Example \ref{ex:1}).

In terms of the convergence rate, another interesting observation in Example \ref{ex:1} is that the Hausdorff distance between $C_{k}$ and the maximal $C_{max}$ decays to $0$ at least exponentially fast for the limit types (i) and (ii). 

Next, we state the main result of this work, which shows that the observations made in Example \ref{ex:1} actually hold for any linear systems of the form \eqref{eqn:sys}:
\begin{theorem} \label{thm:main} 
Under Assumption \ref{asp:cpt}, suppose that $C_0$ is a compact and convex RCIS of the system  $\Sigma$ in $S_{xu}$. If $C_0 $ is contained in the interior of $C_{k_{0}}$ for some $k_0>0$, then $C_{k}$ converges to the maximal RCIS $C_{max}$ in Hausdorff distance, and there exists integers $N_0 \geq 0$, $N>0$ and scalars $c>0$, $a\in (0,1)$ such that the Hausdorff distance between $C_{N_0+kN}$ and $C_{max}$ satisfies 
\begin{align} \label{eqn:d_bound} 
    d(C_{N_0+kN}, C_{max}) \leq c a^{k}.
\end{align}
That is, the Hausdorff distance $d(C_{N_0+kN}, C_{max})$ decays to $0$ exponentially fast as $k$ goes to infinity.
\end{theorem}
Note that the bound in \eqref{eqn:d_bound} only holds for indices $\{N_0+kN\}_{k=0}^{ \infty} $ increasing by an interval of $N$. But, due to the fact that $d(C_{k}, C_{max})$ is monotonically non-increasing over $k$, the inequality in \eqref{eqn:d_bound} implies for all $k\geq 0$,
\begin{align} \label{eqn:d_bound_2} 
    d(C_{N_0+k}, C_{max}) \leq c (a^{1/N})^{k-N+1} = (c a^{-1+1/N} )a^{k/N}. 
\end{align}
Hence, by Theorem \ref{thm:main} and \eqref{eqn:d_bound_2} , the $k$-step backward reachable set $C_{k}$ converge to the maximal RCIS $C_{max}$ in Hausdorff distance exponentially fast whenever $C_0$ is contained in the interior of $C_{k_0}$ for some $k_0>0$. This result validates the two key observations we have in Example \ref{ex:1}.

When $C_0$ and $C_{k}$ are represented by polytopes, the condition $C_0 \subseteq int(C_{k})$ can be numerically checked by a linear program\footnote{This is a variant of the standard polytope containment problem given the $H$-representations of two polytopes\cite{sadraddini2019linear}.}.  Suppose that  $C_0 = \{x \mid H_1 x \leq  h_1\} $ and $C_{k} = \{x \mid H_2 x \leq h_2\} $ for some $H_i$, $h_i$ in appropriate dimensions, $i=1$, $2$, and $x_{0}$ is any interior point of $C_{k}$. We construct the following linear program to check if $C_0 \subseteq int(C_{k})$:
\begin{align}
	\begin{split} \label{eqn:lp} 
\gamma^{*}	= \min_{ \gamma \geq 0, \Lambda} &~ \gamma\\
	\text{subject to } & \Lambda H_1 = H_2 \\
					   & \Lambda (h_1-H_1 x_0) \leq  \gamma (h_2-H_2 x_{0}).
	\end{split}
\end{align}
According to \cite[Lemma 1]{sadraddini2019linear},  $ \gamma^{*}$ in \eqref{eqn:lp} is less than $1$ if and only if $C_0 \subseteq int(C_{k})$ if and only if $k_0 \leq k$. 

Next, recall that $ C_{\infty} $ in \eqref{eqn:C_inf} is the limit of the $k$-step backward reachable set $C_k$ as $k$ goes to infinity. When $k_0$ in Theorem \ref{thm:main} exists, it is obvious that $C_0 $ is contained by the interior $ int(C_{ \infty})$ of $C_{ \infty}$, since $C_{k_0}$ is a subset of $C_{ \infty}$. But conversely, if $C_0 \subseteq int(C_{ \infty})$, does there always  exist a finite $k_0$ such that $C_0 \subseteq int(C_{k_0})$? It turns out that those two conditions are equivalent, shown by the following theorem:
\begin{theorem} \label{thm:C_inf} 
Under the same conditions of Theorem \ref{thm:main}, $C_{0}$ is contained in the interior of $C_{k_0}$ for some finite $k_0>0$ if and only if $C_0$ is contained in the interior of $C_{ \infty}$ in \eqref{eqn:C_inf}. 
\end{theorem}

Finally, the readers may wonder what happens if $k_0$ in Theorem \ref{thm:main}  does not exist, or equivalently $C_0 \not \subseteq int(C_{ \infty})$. In Example \ref{ex:1}, the limit type (iii) occurs when $k_0$ does not exist. However, there are examples where the limit types (i) and (ii) occur even if a finite $k_0$ does not exist. Thus, the existence of $k_0$ is only a sufficient condition for the convergence of the backward reachable set  $C_{k}$ to the maximal RCIS. When the system is disturbance-free (that is $D=\{0\} $), the existence of $k_0$ is guaranteed if $(A,B)$ is asymptotically stable and $0\in int(C_0)$, or if $(A,B)$ is controllable and $0\in C_0$ and $0\in int(S_{xu})$. A more in-depth discussion for results of disturbance-free systems can be found in Section \ref{sec:dis}. 
\section{Proof of the main result} \label{sec:pf} 
In this section, we present the main ideas in the proof of Theorem \ref{thm:main}. 
First, according to a fixed-point theorem \cite[Theorem 12]{caravani2002doubly}\footnote{See Appendix \ref{sec:feq} for a detailed discussion.}, given a compact convex RCIS $C_0$ in $S_{xu}$, there always exists a stationary point $(x_{e},u_{e},d_{e})\in S_{xu}\times D$ such that $x_{e}\in C_0$ and $Ax_{e}+Bu_{e}+Ed_{e}= x_{e}$.
Without loss of generality, we assume that the stationary point $(x_{e},u_{e},d_{e})$ is the origin of the state-input-disturbance space and thus $0\in C_0$ and $0\in S_{xu}\times D$ for the remainder of this section. 
Also, by Theorem \ref{thm:C_inf}, we convert the condition on the existence of $k_0$ in Theorem \ref{thm:main} into the following equivalent assumption.
\begin{assumption} \label{asp:int} 
   The set $C_0$ is a compact convex RCIS of $\Sigma$ in $S_{xu}$ and is contained by the interior $int(C_{ \infty})$ of the infinite-step backward reachable set $C_{ \infty}$ in \eqref{eqn:C_inf}. 
\end{assumption}
Based on Assumption \ref{asp:int}, the proof of Theorem \ref{thm:main} contains two steps: The first step is to show that the closure $cl(C_{ \infty})$ of the limit $C_{ \infty}$ is equal to the maximal RCIS $C_{max}$, that is to prove the following theorem:    
\begin{theorem} \label{thm:C_max} 
	Under Assumptions \ref{asp:cpt} and \ref{asp:int}, the closure of $C_{ \infty}$ in \eqref{eqn:C_inf} is the maximal RCIS $C_{max}$, that is $cl(C_{ \infty}) = C_{max}$.
\end{theorem}
\begin{proof}[Sketch proof]
Note that the maximal RCIS $C_{max}$ is bounded due to Assumption \ref{asp:cpt}. Since $C_0 \subseteq int(C_{ \infty})$ and $C_{max}$ is bounded, intuitively, we can find a small enough $ \alpha >0$ such that the set $C(\alpha) =(1- \alpha)C_{0} +\alpha C_{max}$ is contained in $int(C_{ \infty})$. Due to the linearity of the system, it can be shown that the infinite-step backward reachable set of $C( \alpha)$ contains $(1-\alpha) C_{ \infty} + \alpha C_{max}$. Then, the key to prove Theorem \ref{thm:C_max} is to realize that for any set $C' \subseteq int(C_{ \infty})$, the infinite-step backward reachable set of $C'$ is contained by $C_{ \infty}$. Thus, for some $ \alpha>0$, we have  $(1-\alpha) C_{ \infty} + \alpha C_{max} \subseteq C_{ \infty} \subseteq C_{max}$. That implies $C_{max} = cl(C_{ \infty})$. The complete proof can be found in Appendix.
\end{proof}
The second step is to show that the Hausdorff distance between $C_k$ and and the maximal RCIS $C_{max}$ decays to $0$ exponentially fast when $C_0 \subseteq int(C_{ \infty})$. We first introduce an extended notion of $ \lambda$-contractive sets\cite{blanchini1994ultimate}:
\begin{definition} \label{def:contractive} 
  Given a scalar $ \lambda>0$, a set $X$ is called \emph{$k$-step $ \lambda$-contractive} if the $k$-step backward reachable set of $ \lambda X$ contains the set $X$, that is $Pre_{\Sigma}^{k}( \lambda X, S_{xu},D)\supseteq X$.
\end{definition}
By definition, a set $X$ is $k$-step $\lambda$-contractive if the system can go from any state in $X$  to some state in $\lambda X$ in $k$ steps without violating any safety constraints.
Under Assumptions \ref{asp:cpt} and \ref{asp:int}, the closure of the infinite-step backward reachable set $cl(C_{ \infty})$ always contains a $N$-step $ \lambda$-contractive set, shown by the following theorem. 
\begin{theorem} \label{thm:conv_inv_ext} 
	Under Assumptions \ref{asp:cpt} and \ref{asp:int}, there exist an integer $N>0$ and  scalars ${\gamma\in (0,1]}$, $ \lambda\in[0,1)$ such that $ \gamma cl(C_{ \infty})$ is $N$-step $ \lambda$-contractive. Furthermore, there exists a finite integer $N_0 \geq 0$ such that $Pre_{\Sigma}^{N_0}(C_0, S_{xu},D) \supseteq  \lambda \gamma cl(C_{ \infty})$. 
\end{theorem}
\begin{proof}[Sketch proof]
	Since $ C_0 \subseteq int(C_{ \infty})$ and $C_{ \infty}$ is convex, there exists positive scalars $ \beta_0$ and $ \beta_1$, with $0 < \beta_0 < \beta_1 <1$, such that $C_0 \subseteq \beta_0 cl(C_{ \infty}) \subset \beta_1 cl(C_{ \infty}) \subseteq int(C_{ \infty})$. Since $ \beta_1 cl(C_{ \infty}) \subseteq int(C_{ \infty})$ and the $k$-step backward reachable set of $C_0$ converges to $C_{ \infty}$, it can be proven that there exists a finite $N$ such that $ \beta_1 cl(C_{ \infty}) \subseteq C_{N} $. Since $C_0 \subseteq \beta_0 cl(C_{ \infty})$, $C_{N}$ is contained by the $N$-step backward reachable set of $ \beta_0 cl(C_{ \infty})$. That implies $ \beta_1 cl(C_{ \infty}) \subseteq C_{N}$ is contained in the $N$-step backward reachable set of $\beta_0 cl(C_{ \infty})$, and thus is a $N$-step $(\beta_0/\beta_1)$-contractive set. By assigning $ \gamma = \beta_1$, $\lambda = \beta_0/ \beta_1$ and $N_0=N$, the statement in Theorem \ref{thm:conv_inv_ext} is proven.  A complete proof can be found in Appendix.
\end{proof}
By Theorems \ref{thm:C_max} and \ref{thm:conv_inv_ext}, we show that $ \gamma cl(C_{ \infty})= \gamma C_{max}$ is $N$-step $ \lambda$-contractive. Note that a $k$-step $ \lambda$-contractive set $X$ is not necessarily an RCIS unless $k=1$ and $0\in X$. Thus, $ \gamma C_{max}$ may not be an RCIS.

Recall that our goal in the second step of proving Theorem \ref{thm:main} is to show the convergence rate of the $k$-step backward reachable set $C_{k}$ to $C_{max}$. So how is $ \gamma C_{max}$ being $k$-step $\lambda$-contractive set related to the convergence rate of backward reachable sets? Let $C$ be an RCIS of $\Sigma$ in $S_{xu}$. 
It turns out that if there exists a factor $\gamma\in (0,1]$ such that the scaled set $\gamma C$ is $N$-step $ \lambda$-contractive for some $N$ and $ \lambda\in [0,1)$, then the $k$-step backward reachable set of $ \lambda\gamma C$ approaches to $C$ exponentially fast as $k$ increases. The proof of this statement is enabled by the following theorem.
\begin{theorem} \label{thm:conv_ext} 
Under Assumption \ref{asp:cpt}, for any convex RCIS  $C$ of $\Sigma$ in $S_{xu}$, suppose that there exist $ \gamma\in (0,1)$, $N$ and $ \lambda\in[0,1)$ such that $ \gamma C$ is  $N$-step $ \lambda$-contractive, that is
\begin{align}
    Pre_{\Sigma}^{N}( \lambda \gamma C, S_{xu},D) \supseteq \gamma C.
\end{align}
Then, for any scalar  $ \xi $ with $ 1 > \xi \geq \lambda \gamma$,
\begin{align} \label{eqn:conv_1stp}  
   Pre_{\Sigma}^{N}( \xi C, S_{xu}, D) \supseteq g( \xi) C, 
\end{align}
where 
\begin{align} \label{eqn:g} 
    g( \xi) = \frac{1- \gamma}{1- \lambda \gamma} \xi + \frac{(1- \lambda) \gamma}{1- \lambda \gamma} \geq  \xi. 
\end{align}
\end{theorem}
For any $C$ satisfying the conditions in Theorem \ref{thm:conv_ext}, \eqref{eqn:conv_1stp} implies that the $N$-step backward reachable set of $ \lambda \gamma C$ expands at least to $g(\lambda \gamma) C \supseteq \lambda\gamma C$. By applying Theorem \ref{thm:conv_ext} $k$ times, we obtain that
\begin{align} \label{eqn:conv_kstp} 
   Pre_{\Sigma}^{kN}(\lambda \gamma C, S_{xu}, D) \supseteq g^{k}( \lambda \gamma) C, 
\end{align}
where $g^{k}(\cdot)= g(g(...))$ is the function that composes $g(\cdot)$ for $k$ times.
According to Theorems \ref{thm:C_max} and \ref{thm:conv_inv_ext}, $C$ in \eqref{eqn:conv_kstp} can be replaced by $C_{ max}$. That is, for $\gamma$, $N$ and $ \lambda$ in Theorem \ref{thm:conv_inv_ext}, 
\begin{align} \label{eqn:conv_kstp_2} 
Pre_{\Sigma}^{kN}( \lambda \gamma C_{max}, S_{xu}, D) \supseteq g^{k}( \lambda \gamma) C_{max}. 
\end{align}
By Theorem \ref{thm:conv_inv_ext}, there exists $N_0$ such that the $N_0$-step backward reachable set $C_{N_0} $ of $C_0$ contains $\lambda \gamma C_{max}$.  
Then, by \eqref{eqn:conv_kstp_2}, 
\begin{align} \label{eqn:conv_kstp_3} 
	\begin{split}
C_{max} \supseteq C_{N_0+kN} \supseteq Pre_{\Sigma}^{kN}( \lambda \gamma C_{max}, S_{xu}, D) \supseteq g^{k}(\lambda \gamma ) C_{max}. 
	\end{split}
\end{align}
The inclusion relation in \eqref{eqn:conv_kstp_3} implies that the Hausdorff distance between  $C_{N_0+kN}$  and the maximal RCIS $C_{max}$ is bounded by
\begin{align} \label{eqn:bound_hd} 
   d(C_{N_0+kN},C_{max})  \leq (1-g^{k}( \lambda \gamma)) \sup_{x\in C_{max}} \Vert x\Vert_{2}.
\end{align}
Since $g(\cdot)$ in \eqref{eqn:g} is affine, it can be shown that  
 \begin{align} \label{eqn:gk} 
     1- g^{k}(\lambda \gamma) = \left( \frac{1- \gamma}{1- \lambda \gamma} \right)^{k} (1- \lambda \gamma).
 \end{align}
Combining \eqref{eqn:bound_hd} and \eqref{eqn:gk}, we have   $$d(C_{N_0+kN}, C_{max}) \leq c a^{k},$$ 
where $c = (1-  \lambda\gamma)\sup_{x\in C_{max}} \Vert x\Vert_{2} $ and ${a= (1- \gamma) / (1- \lambda \gamma)}$. With $ \gamma \in (0,1]$, $ \lambda\in [0,1)$, it is easy to check that $a \in [0,1)$. Since $C_{max}$ is bounded, $c$ is finite.  Thus, the Hausdorff distance between $C_{N_0+kN}$ and $C_{max}$ decays to $0$ exponentially fast as $k$ goes to infinity. That completes the proof of Theorem \ref{thm:main}.   

\section{Numerical Example} \label{sec:example} 
Consider the two-dimensional system $\Sigma$
\begin{align}
    \Sigma: x(t+1) = \begin{bmatrix}
    1.1 & 1 \\ 0 & 1
    \end{bmatrix}x + 
	\begin{bmatrix}
	0 \\ 1	
	\end{bmatrix} u +
	\begin{bmatrix}
	1\\1	
	\end{bmatrix} d,
\end{align}
with $x\in \R^{2}$, $u\in \R$ and $d\in D=[-0.01,0.01]$. The safe set of the system is $S_{xu} = S_{x}\times U$, with $S_{x} = [-4,4]\times [-2,2]$ and $U=[-0.3,0.3] $. Denote the maximal RCIS of $\Sigma$ in $S_{xu}$ by $C_{max}$. The set $C_0$ is selected to be the maximal RCIS in a scaled safe set $\tilde{S}_{xu} =(0.1S_{x})\times U$, shown by the yellow polytope in Fig. \ref{fig:visual}.   

MPT3 toolbox\cite{MPT3} and YALMIP \cite{Lofberg2004}, equipped with GUROBI 9.5.0\cite{gurobi}, are used to implement the computation of the backward reachable set in \eqref{eqn:pre} and the linear program in \eqref{eqn:lp}. The computed $k$-step backward reachable sets $C_{k}$ of $C_0$ for $k=1, \cdots, 16$ are visualized in Fig. \ref{fig:visual}, which converges to the maximal $C_{max}$ in $16$ steps. By solving the linear program in \eqref{eqn:lp}, we checked that $C_0$ is contained in the interior of $C_2$. Thus, Theorem \ref{thm:main} implies that  the $k$-step backward reachable set $C_k$ converges to the maximal RCIS $C_{ max}$ at least exponential fast. The Hausdorff distance between $C_{k}$ and $C_{max}$ is shown by the red curve in Fig. \ref{fig:conv_curve}. We also manually fit an exponential function $y(k)=4.61\times 0.8^{k}$ (the blue curve in Fig. \ref{fig:conv_curve}) that bounds the actual Hausdorff distance from above. The existence of this exponential decaying upper bound is predicted by Theorem \ref{thm:main}. Given any system and $C_0$, how to find such an exponential decaying function without computing all the backward reachable sets would be part of our future work.

\begin{figure}[]
	\centering
	\includegraphics[width=0.45\textwidth]{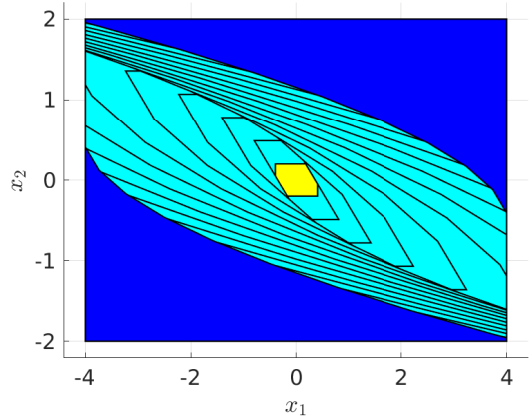}
	\caption{The $k$-step backward reachable set $C_{k}$ for $k=1, \cdots, 16$ are the nested cyan polytopes, where a larger one corresponds to a larger $k$. The set $C_0$ is the yellow polytope in the middle. The maximal RCIS $C_{max}$ is equal to the largest cyan polytope $C_{16}$. The dark blue rectangle is the safe set $S_{xu}$.}
	\label{fig:visual}
\end{figure}

\begin{figure}[]
	\centering
	\includegraphics[width=0.38\textwidth]{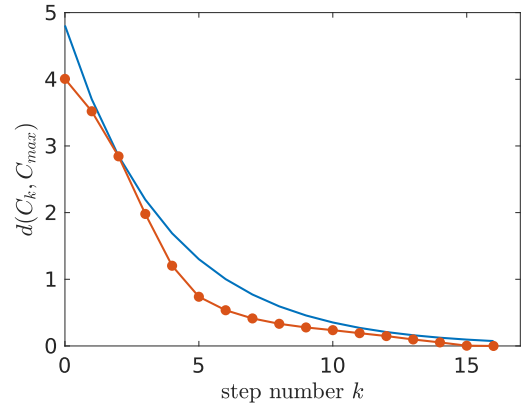}
	\caption{The Hausdorff distance $d(C_k, C_{max})$ between the $k$-step backward reachable set $C_{k}$ and the maximal RCIS $C_{max}$ versus the number of steps $k$ (the red curve with dots). An exponential function $y(k) = 4.61\times 0.8^k$  that bounds the Hausdorff distance from above is manually fitted (the blue curve).}
	\label{fig:conv_curve}
\end{figure}
\section{Discussion} \label{sec:dis} 
In this section, we compare our result with the existing ones in \cite{de2004computation, darup2014general,gutman1987algorithm,cwikel1986convergence}.
Here we adopt the notation in \cite{de2004computation, darup2014general} and call a set $X$ a \emph{C-set} if $X$ is convex, compact and contains $0$ in the interior. Note that the results in \cite{de2004computation, darup2014general,gutman1987algorithm,cwikel1986convergence} are all based on the condition that the disturbance set $D=\{0\} $ and $S_{xu}$ is a C-set. Hence, we assume that the above condition holds for the remainder of this section. 

First, \cite{de2004computation, darup2014general,gutman1987algorithm} identify the two sufficient conditions under which the $k$-step backward reachable set $C_{k}$ converges to the maximal CIS in Hausdorff distance.  The first sufficient condition \cite{gutman1987algorithm, de2004computation,darup2014general} is that $(A,B)$ is controllable and $C_0=\{0\} $; the second one \cite[Proposition 30]{de2004computation} is that $(A,B)$ is asymptotically stabilizable and $C_0$ is a controlled invariant C-set. When the first sufficient condition is satisfied, \cite{de2004computation, darup2014general} show that $C_0=\{0\}$ also satisfies the conditions in Theorem \ref{thm:main}. When the second sufficient condition is satisfied, the corresponding $C_0$ may not satisfy the condition in Theorem \ref{thm:main}. However, the proof of \cite[Proposition 30]{de2004computation} shows that $C_0$ must contain a smaller CIS that satisfies the conditions in Theorem \ref{thm:main}. Thus, both of the sufficient conditions are corollaries of our result.

Though \cite{de2004computation, darup2014general,gutman1987algorithm} prove the convergence of $C_{k}$, none of those works show the convergence rate of $C_{k}$.  To our best  knowledge, \cite{cwikel1986convergence} is the only work that shows the convergence rate of $C_{k}$, under the condition that $(A,B)$ is controllable and $C_0=\{0\} $ (the first sufficient condition above). Originally \cite{cwikel1986convergence} derives the convergence rate in a different metric. But by converting the metric in \cite{cwikel1986convergence} to Hausdorff distance, we can show that the convergence rate of $C_{k}$ in \cite{cwikel1986convergence} is equivalent to the exponential convergence in Hausdorff distance proven in this work. That is, under the first sufficient condition in the previous paragraph, our result coincides with the result in \cite{cwikel1986convergence}.

\section{Conclusion} \label{sec:con} 
In this work, we consider linear systems with disturbance and show a sufficient condition under which the $k$-step backward reachable set of RCISs converges to the maximal RCIS exponentially fast. When all sets are represented by polytopes, this sufficient condition can be numerically checked via the linear program in \eqref{eqn:lp}.  When restricted to disturbance-free systems, our result implies the existing results in  \cite{de2004computation, darup2014general,gutman1987algorithm,cwikel1986convergence}.

In terms of applications, our result provides convergence guarantees for inside-out algorithm (and its variants) \cite{gutman1987algorithm,de2004computation,rungger2017computing,fiacchini2017computing,wintenberg2020implicit, anevlavis2021enhanced}, and sheds some light on novel analysis for RCIS-based control synthesis algorithms. 
For instance, in constrained model predictive control (MPC) equipped with a controlled invariant terminal set \cite{blanchini2008set}, our results imply that under mild conditions, the DoA of MPC enlarges to the maximal DoA exponentially fast as the prediction horizon $T$ increases, which provides new insights for selecting the prediction horizon $T$. Moreover, if parameters $N_0$, $N$, $a$ and $c$ in Theorem \ref{thm:main} are known, we can quantitatively evaluate the Hausdorff distance between the DoA of MPC with respect to any given $T$ and the maximal DoA via \eqref{eqn:d_bound}.  The algorithms for estimating parameters $N_0$, $N$, $a$ in Theorem \ref{thm:main} will be part of our future work. 
\bibliographystyle{IEEEtran}
\bibliography{ref}

\appendix

\subsection{On the existence of stationary points} \label{sec:feq} 
\cite[Theorem 12]{caravani2002doubly} only considers disturbance-free systems with decoupled safe set $S_{xu}= X\times U$ for a convex compact $X \subseteq \R^{n}$ and a closed convex $U \subseteq \R^{m}$. To apply \cite[Theorem 12]{caravani2002doubly} to our setting, given the system $\Sigma$ of the form \eqref{eqn:sys}, we construct the disturbance-free system $\Sigma'$:
\begin{align}
   \begin{bmatrix}
   x(t+1)\\u(t+1)	
   \end{bmatrix} = 
   \begin{bmatrix}
   A & B \\ 0 & 0
   \end{bmatrix} \begin{bmatrix}
   x(t)\\u(t)
   \end{bmatrix} + 
   \begin{bmatrix}
   0\\ I	
   \end{bmatrix} v_1(t) +
	\begin{bmatrix}
	E \\0	
	\end{bmatrix} v_2(t),
\end{align}
where $(x,u)$ is the state and $v_1$, $v_2$ are two inputs of the system $\Sigma'$. Given a compact convex RCIS $C_{0}$ of $\Sigma$ in $S_{xu}$, construct the set $C_0'$ as
\begin{align}
    C_0' = \{(x,u) \in S_{xu} \mid x\in C_0, \exists d\in D, Ax+Bu+Ed \in C_0\}. 
\end{align}
It is easy to check that $C_0'$ is a compact convex controlled invariant set\cite{caravani2002doubly} of the system $\Sigma'$ in the decoupled safe set $S_{xu}\times \R^{m}\times D$. Then, by \cite[Theorem 12]{caravani2002doubly}, there exists $(x_{e},u_{e})\in C_{0}'$, $v_{1,e}\in \R^{m}$ and $v_{2,e}\in D$ such that 
\begin{align} \label{eqn:feq} 
    \begin{bmatrix}
    x_{e}\\u_{e}	
    \end{bmatrix} = 
	\begin{bmatrix}
	A  &  B \\ 0 & 0
    \end{bmatrix}
	\begin{bmatrix}
	x_{e}\\ u_{e}
	\end{bmatrix} + 
	\begin{bmatrix}
	0 \\ I	
	\end{bmatrix} v_{1,e} + 
	\begin{bmatrix}
	E \\0	
	\end{bmatrix} v_{2,e}.
\end{align}
Define $d_{e}= v_{2,e}$. By \eqref{eqn:feq}, $x_{e}=Ax_{e} + Bu_{e}+Ed_{e}$. By construction of $C_0'$, we have $x_{e}\in C_0$ and  $(x_{e},u_{e})\in S_{xu}$. 

\subsection{Proofs of Theorems \ref{thm:C_inf}-\ref{thm:conv_ext}} \label{sec:pf_thms} 
The following lemmas reveal properties of the backward reachable sets, which are crucial for proving Theorems \ref{thm:C_inf}-\ref{thm:conv_ext}.
\begin{lemma} \label{lem:pre_decomp} 
	For a linear system $\Sigma$ in form of \eqref{eqn:sys}, 
\begin{align} \label{eqn:pre_decomp} 
	\begin{split}
Pre_{\Sigma}(X_1+X_2, S_{xu,1}+S_{xu,2},D_1+D_2) \supseteq  \\
Pre_{\Sigma}(X_1, S_{xu,1},D_1) + Pre_{\Sigma}(X_2, S_{xu,2}, D_2).
	\end{split}
\end{align}
\end{lemma}
\begin{proof}
Let $x_i\in Pre_{\Sigma}(X_{i}, S_{xu,i},D_{i})$ for $i=1,2$. Then, there exists $u_i$ such that $(x_{i},u_{i})\in S_{xu,i}$ and $ Ax_{i}+Bu_{i}+ED_{i} \subseteq X_{i}$, for $i=1,2$. Thus, we have
\begin{align}
	(x_1,u_1)+(x_2,u_2)\in S_{xu,1}+S_{xu,2}\\
	A(x_1+x_2)+B(u_1+u_2)+ED_1+ED_2 \subseteq X_1 + X_2.
\end{align}
Thus, $x_1+x_2\in Pre_{\Sigma}(X_1+X_2, S_{xu,1}+S_{xu,2}, D_1+D_2)$.
\end{proof}

\begin{lemma} \label{lem:pre_decomp_cvx} 
	For a linear system $\Sigma$ in form of \eqref{eqn:sys}, if $X$, $S_{xu}$ and $D$ are convex, then for any $a\in[0,1]$, $b\in[0,1]$ and $k\geq 1$, we have
\begin{align} \label{eqn:pre_decomp_cvx} 
		Pre^{k}_{\Sigma}(X, S_{xu},D) \supseteq & Pre^{k}_{\Sigma}(aX, bS_{xu},bD) +\\
& Pre^{k}_{\Sigma}((1-a)X, (1-b)S_{xu}, (1-b)D).\nonumber
\end{align}
\end{lemma}
\begin{proof}
For any convex set $C$, $C = \alpha C+(1- \alpha)C$ for any $ \alpha\in [0,1]$. Thus, by Lemma \ref{lem:pre_decomp}, 
\begin{align}
	Pre_{\Sigma}(X, S_{xu},D) \supseteq &  Pre_{\Sigma}(aX, bS_{xu},bD) + \\
 &Pre_{\Sigma}((1-a)X, (1-b)S_{xu}, (1-b)D). \nonumber  
\end{align}
Now assume that \eqref{eqn:pre_decomp_cvx} holds for $k=N-1$. By applying Lemma \ref{lem:pre_decomp} again, it can be easily proven that \eqref{eqn:pre_decomp_cvx} holds for $k=N$.  
Thus, by induction that \eqref{eqn:pre_decomp_cvx} holds for all $k$.
\end{proof}

\begin{lemma} \label{lem:pre_linear} 
   For a linear system $\Sigma$ in form of \eqref{eqn:sys} and any $a \geq 0$,
   \begin{align}
       Pre_{\Sigma}(aX,a S_{xu},a D) = a Pre_{\Sigma}(X,S_{xu},D).
   \end{align}
\end{lemma}
\begin{proof}
Let $x\in Pre_{\Sigma}(X,S_{xu},D)$. Then, there exists $u$ such that $(x,u)\in S_{xu}$ and $Ax+Bu+ED \subseteq X$. Thus, $(ax,au)\in a S_{xu}$ and $A ax + Bau+ aED \subseteq a X$. That is, $ax\in Pre_{\Sigma}(aX, aS_{xu}, aD)$. Therefore, $a Pre_{\Sigma}(X,S_{xu},D) \subseteq Pre_{\Sigma}(aX, aS_{xu},aD)$.

Next, let $ax\in Pre_{\Sigma}(aX,aS_{xu},aD)$. Then, there exists $au$ such that $(ax,au)\in aS_{xu}$ and $Aax+Bau+aED \subseteq aX$. Thus, $(x,u)\in S_{xu}$ and $Ax + Bu+ ED \subseteq X$. That is, $x\in Pre_{\Sigma}(X, S_{xu}, D)$. Therefore, $Pre_{\Sigma}(aX, aS_{xu},aD) \subseteq a Pre_{\Sigma}(X,S_{xu},D)$. 
\end{proof}

\begin{lemma} \label{lem:h_conv} 
 Let $\{C_{i}\}_{i=0}^{ \infty}$ be an expanding family of nonempty  compact convex set. That is, $C_{i} \subseteq C_{j}$ for any $i$, $j$ such that $i \leq j$. Suppose that $int(C_{i}) \not= \emptyset$ for all $i \geq  i_0$ for some $i_0\geq 0$, and $C_{ \infty} = \cup_{i=0}^{ \infty} C_{i}$ is bounded. Then, for any closed set $U \subseteq int(C_{ \infty})$, there  exists $i$ such that $C_{i} \supseteq U$.
\end{lemma}
\begin{proof}
Without loss of generality, we assume $i_0 =0$. Since $C_{i}$ is convex and compact, it is easy to check that $cl(int(C_{i})) = C_{i}$. 

Define $ \tilde{C}_{ \infty} = \lim_{k \to \infty} \cup_{i=0}^{k} int(C_{i})$. Thus, $ \tilde{C}_{ \infty}$ is an open convex set. Note that 
\begin{align}
   cl(C_{ \infty}) \supseteq cl( \tilde{C}_{ \infty}) \supseteq cl( int(C_{i}))  = C_{i}, \forall i \geq 0.
\end{align}
Thus,
\begin{align}
   cl(C_{ \infty}) \supseteq cl( \tilde{C}_{ \infty}) \supseteq \lim_{k \to \infty} \cup_{i=0}^{k} C_{i}= C_{ \infty}.
\end{align}
Thus, $cl(C_{ \infty})$ is the closure of $\tilde{C}_{ \infty}$. By convexity of $ C_{ \infty}$, $\tilde{C}_{ \infty}$ is the interior of $cl(C_{ \infty})$, and also the interior of $C_{ \infty}$. 

Next, let $U$ be a closed subset of $ int(C_{ \infty})= \tilde{C}_{ \infty}$. Since $C_{ \infty}$ is bounded, $U$ is compact. Also, since $ U \subseteq \tilde{C}_{ \infty}= \cup_{i=0}^{ \infty}int(C_{i}) $, $ \{ int(C_{i})\}_{i=0}^{ \infty} $ forms a open cover of $U$. By compactness, there exists a finite subcover $\{C_{i_{k}}\}_{k=0}^{K} $ of $U$ for some $K\geq 0$. Then, $U \subseteq int(C_{i_{K}})$.
\end{proof}

\begin{lemma} \label{lem:pre_cpt} 
    Let $X$ be a nonempty compact convex set in $\R^{n}$. Under Assumption \ref{asp:cpt}, $Pre_{\Sigma}(X, S_{xu}, D)$ is compact and convex. 
\end{lemma}
\begin{proof}
A proof for similar results can be found in \cite{blanchini1994ultimate}. We provide a separate proof here for completeness.

Define $C_{xu} = \{(x,u)\in S_{xu} \mid Ax+Bu+ED \subseteq X\} $.  By definition of backward reachable set, $Pre_{\Sigma}(X, S_{xu}, D) = \pi_{[1,n]}(C_{xu})$. For now, we assume that $C_{xu}$ is compact and convex. Since projection $\pi_{[1,n]}(\cdot)$ is continuous and $C_{xu}$ is compact, $Pre_{\Sigma}(X, S_{xu}, D)$ is compact. Since the projection of a convex set is convex, $Pre_{\Sigma}(X, S_{xu}, D)$ is convex. It is left to show the convexity and compactness of $C_{xu}$.

We first show that $C_{ xu}$ is convex. Let $(x_1,u_1)$, $(x_2,u_2)$ be two points in $C_{xu}$ and $ \alpha$ be a constant in $(0,1)$. Denote $(x,u) = \alpha (x_1,u_1)+ (1- \alpha)(x_2,u_2)$. Since $S_{xu}$ is convex, $(x,u)$ is in $S_{xu}$. For any $d\in D$ and $i\in \{1,2\} $, $x_{i}^{+}=Ax_{i}+Bu_{i}+Ed\in X$. Since $X$ is convex, for the same $d\in D$, $Ax+Bu+Ed= \alpha x_{1}^{+} +(1- \alpha)x_2^{+} $ is in $X$. Thus, $(x,u)\in C_{xu}$.  That is, $C_{xu}$ is convex.

Next, we show that $C_{ xu}$ is compact. Let $\{(x_{n},u_{n})\}_{n=1}^{ \infty} $ be an arbitrary convergent sequence in $C_{xu}$. Suppose that $(x_{n},u_{n}) \rightarrow (x,u)$. Since $S_{xu}$ is compact, $(x,u)\in S_{xu}$. For any $d\in D$, $Ax_{n}+Bu_{n}+Ed \rightarrow Ax+Bu+Ed$. Since $X$ is compact and $Ax_{n}+Bu_{n}+Ed\in X$, $Ax+Bu+Ed\in X$. Since $(x,u)\in S_{xu}$ and $Ax+Bu+ED  \subseteq X$, $(x,u)\in C_{xu}$. Thus, $C_{xu}$ is closed. Since $C_{xu}$ is closed subset of the compact set $S_{xu}$, $C_{xu}$ is compact. 

\end{proof}

\begin{lemma} \label{lem:c_inf} 
    Suppose Assumption \ref{asp:cpt} holds. Let $C_0$ be a compact convex RCIS of $\Sigma$ in $S_{xu}$.  Define $C_{k} = Pre_{\Sigma}^{k}(C_0, S_{xu},D)$ and $C_{ \infty}= \cup_{k=1}^{ \infty} C_{k}$. Then, $C_{ \infty}$ satisfies the following properties:
	\begin{itemize}
		\item [(a)] $\{C_{k}\}_{k=1}^{ \infty} $ is an expanding family of compact convex RCISs;
		\item[(b)] $C_{ \infty}$ is  bounded and convex;
		\item[(c)] If $int(C_{ \infty})$ is nonempty, then there exists a finite $k\geq 0$ such that $int(C_{k})$ is nonempty;
		\item[(d)] For any compact subset $C$ contained in the interior of $C_{  \infty}$, there exists an $ \epsilon>0$ such that $C + B_{ \epsilon}(0) \subseteq C_{ \infty}$.
	\end{itemize}
\end{lemma}
\begin{proof}
	We first prove (a): Since $C_0$ is an RCIS in $S_{xu}$, by definition, $ { C_0 \subseteq Pre_{\Sigma}(C_0, S_{xu}, D) = C_1}$. Since $C_0$ is compact and convex, by Lemma \ref{lem:pre_cpt},   $C_{1}$ is compact and convex. Now suppose that $C_{k-1} \subseteq C_{k}$, and both $C_{k-1}$ and $C_{k}$ are compact and convex. Then, $C_{k} = Pre_{\Sigma}(C_{k-1},S_{xu},D) \subseteq Pre_{\Sigma}(C_{k}, S_{xu},D) = C_{k+1}$.  Since $C_{k}$ is compact and convex, by Lemma \ref{lem:pre_cpt}, $C_{k+1}$ is compact and convex. By induction, $\{C_{k}\}_{k=1}^{ \infty} $ is an expanding family of nonempty compact convex sets. Furthermore, since $C_{k} \subseteq C_{k+1} = Pre_{\Sigma}(C_{k}, S_{xu},D)$, $C_{k}$ is an RCIS of $\Sigma$ in $S_{xu}$ for all $k\geq 0$.

	(b) Since $S_{xu}$ is bounded and $C_{k} \subseteq \pi_{[1,n]} (S_{xu})$ for all $k \geq 0$, $C_{ \infty}$ is bounded. Let $x_1$ and $x_2$ be two points in $C_{ \infty}$. Since $\{C_{k}\}_{k=1}^{ \infty} $ is an expanding family of sets and $C_{ \infty} = \cup_{k=1}^{ \infty} C_{k}$, there exists $k_0$ such that $x_1, x_2 \in C_{k_0}$. According to (a), $C_{k_0}$ is convex. Thus, any convex combination of $x_1$ and $x_2$ is in $C_{k_0} \subseteq C_{ \infty}$. Hence, $C_{ \infty}$ is convex.

	(c) Since $int(C_{ \infty})$ is nonempty, we can fit a small hypercube in the interior of $C_{ \infty}$. By definition of $C_{ \infty}$, each vertex of this hypercube is contained by $C_{k}$ for some finite $k$. Since the hypbercube has finitely many vertices, there exists a finite $k_0$ such that $C_{k_0}$ contains all the vertices of the hypercube. Since $C_{k_0}$ is convex, $C_{k_0}$ contains the hypercube and thus has nonempty interior.

	(d) Suppose that for all $ \epsilon >0$, $ C + B_{ \epsilon}(0) \not \subseteq C_{ \infty}$. Then there exists $x_{n}\in C$ such that $B_{ 1/n}(x_{n}) \not \subseteq C_{ \infty}$. Since $C$ is compact, there exists a convergent subsequence $(x_{n_{i}})_{i=1}^{ \infty}$ such that $ x_{ n_{i}} \rightarrow x$ for some $x\in C$. Since $C$ is contained by the interior of $C_{ \infty}$, there exists a constant  $ \epsilon' > 0$ such that $B_{ \epsilon'}(x) \subseteq C_{ \infty}$. Since $n_{i} \rightarrow \infty	$ and $x_{n_{i}} \rightarrow x$, there exists $i'$ such that $1/n_{i'} < \epsilon'/2$ and $ \Vert x_{n_{i'}}-x\Vert_{2} < \epsilon'/2$. Thus, $B_{1/ n_{i'}}(x_{n_{i'}}) \subseteq B_{ \epsilon'}(x) \subseteq C_{ \infty}$. However, by construction, $B_{ 1/n_{i'}}(x_{n_{i'}})\not  \subseteq C_{ \infty}$, contradiction! Thus, there exists $ \epsilon > 0$ such that $C + B_{ \epsilon}(0) \subseteq C_{ \infty}$.
\end{proof}

\begin{proof}[\textbf{Proof of Theorem \ref{thm:C_inf}}]
If $C_{0}$ is contained in the interior of $C_{L_0}$ for some $L_0>0$, it is trivial that $C_0$ is contained in the interior of $C_{ \infty}$. Now suppose that $C_0$ is contained in the interior of $C_{ \infty}$. By Lemma \ref{lem:c_inf} (d), there exists $ \epsilon>0$ such that $cl(C_0+ B_{ \epsilon}(0)) \subseteq int(C_{ \infty})$.   Also, by Lemma \ref{lem:c_inf} (c), there exists $k_0$ such that  $int(C_{k_0})$ is nonempty. Thus by Lemma \ref{lem:h_conv}, there exists $k$  such that $C_{k} \supseteq cl(C_0+B_{ \epsilon}(0))$. For the same $k$, it is clear that  $C_0 \subseteq int(C_{k})$.
\end{proof}

\begin{proof}[\textbf{Proof of Theorem \ref{thm:C_max}}]
We first show that there exists a constant $ \alpha \in (0,1)$ such that $ \alpha C_{max} + (1- \alpha)C_0$ is contained by the interior of $C_{ \infty}$. Since $C_{0}$ is contained in the interior of $C_{ \infty}$, by Lemma \ref{lem:c_inf}, there exists $ \epsilon>0$ such that $C_0+ B_{ \epsilon}(0) \subseteq C_{ \infty}$. Then, $C_{0}+ B_{ \epsilon /2 }(0)$ is in the interior of $C_{ \infty}$. Since $S_{xu}$ is compact, $C_{max}$ is bounded, and thus there exists $ \alpha>0$ such that $ \alpha C_{max} \subseteq B_{ \epsilon/ 2}(0)$. Thus, $ C_0 + \alpha C_{max} \subseteq int(C_{ \infty})$. Since $0\in C_{0}$, $ (1- \alpha) C_{0}+ \alpha C_{max} \subseteq C_0+ \alpha C_{max} \subseteq int(C_{ \infty})$.

Next, it can be easily shown that for any RCIS $C$ in the compact convex safe set $S_{xu}$, the convex hull of $C$ and the closure of $C$ are also RCIS in $S_{xu}$ (a special case is proven by \cite[Proposition 20]{de2004computation}). Thus, the maximal RCIS $C_{max}$ is compact and convex. Since both $C_{max}$ and $C_0$ are compact and convex, $ \alpha C_{max} + (1- \alpha) C_0$ is compact and convex.   By Lemmas \ref{lem:h_conv} and \ref{lem:c_inf}, since $ \alpha C_{max}+ (1- \alpha) C_0$ is a closed subset of $int(C_{ \infty})$, there exists $k_0$ such that $\alpha C_{max}+ (1- \alpha) C_0 \subseteq C_{k_0}$.  

We denote $\overline{C}_{0} = \alpha C_{max}+ (1- \alpha) C_0$. Since $C_{max}$ and $C_0$ are RCISs, it is easy to check that $\overline{C}_{0}$ is also an RCIS in $S_{xu}$. Define $\overline{C}_{k} = Pre_{\Sigma}^{k}(\overline{C}_{0}, S_{xu}, D)$.  Then, by Lemmas \ref{lem:pre_decomp} and \ref{lem:pre_linear}, 
\begin{align}
	\overline{C}_{k}=	Pre_{\Sigma}^{k}(\overline{C}_{0}, S_{xu}, D) \supseteq &\alpha Pre_{\Sigma}^{k}( C_{max}, S_{xu}, D) +\nonumber\\
																  &(1- \alpha) Pre_{\Sigma}^{k}( C_0, S_{xu}, D)\\
	=& \alpha C_{max} + (1- \alpha) Pre_{\Sigma}^{k}( C_0, S_{xu}, D)\nonumber\\
	=& \alpha C_{max} + (1- \alpha)C_{k} 
\end{align}
Define $\overline{C}_{ \infty} = \cup_{k=0}^{ \infty} \overline{C}_{k}$. Then,
\begin{align}
	\overline{C}_{ \infty} &\supseteq \cup_{k=0}^{ \infty}(\alpha C_{max} + (1- \alpha)C_{k}) \\
	&= \alpha C_{max} + (1- \alpha) \cup_{k=0}^{ \infty} C_{k} \\
	&= \alpha C_{max} + (1- \alpha) C_{ \infty}. \label{eqn:pf_3_1} 
\end{align}

Since by construction of $k_0$,  $\overline{C}_{0} \subseteq C_{k_0}$,
\begin{align}
   \overline{C}_{ \infty} \subseteq \cup_{k=1}^{ \infty}Pre_{\Sigma}^{k}(C_{k_0}, S_{xu},D) = \cup_{k=1}^{ \infty}C_{k_0+k}= C_{ \infty}. \label{eqn:pf_3_2} 
\end{align}
Thus, by \eqref{eqn:pf_3_1} and \eqref{eqn:pf_3_2},  
\begin{align}
	&\alpha C_{max} + (1- \alpha) C_{ \infty} \subseteq \overline{C}_{ \infty} \subseteq C_{ \infty}\\
	&cl(\alpha C_{max} + (1- \alpha) C_{ \infty}) \subseteq cl(\overline{C}_{ \infty}) \subseteq cl(C_{ \infty}). \label{eqn:pf_3_3} 
\end{align}
Since $C_{max}$ and $cl(C_{ \infty})$ are compact, $\alpha C_{max} + (1- \alpha) cl(C_{ \infty}) $ is compact. Thus, $$ cl( \alpha C_{max} + (1- \alpha) C_{ \infty}) \subseteq \alpha C_{max} + (1- \alpha)cl(C_{ \infty}).$$ Let $x\in \alpha C_{max} + (1- \alpha) cl(C_{ \infty})$. Then, there exist points $x_1\in C_{max}$ and $x_{2}\in cl(C_{ \infty})$ such that  $x = \alpha x_1 + (1- \alpha) x_2$. Since $x_2 \in cl(C_{ \infty})$, there exists $\{x_{2,k}\}_{k=1}^{ \infty} \subseteq C_{ \infty}$ such that $x_{2,k} \rightarrow x_2$. Since $ \alpha x_1 + (1- \alpha) x_{2,k}\in \alpha C_{max}+(1- \alpha)C_{ \infty}$ and ${\alpha x_1 + (1- \alpha) x_{2,k} \rightarrow x}$, $x\in cl( \alpha C_{max} + (1- \alpha) C_{ \infty})$ and thus $ cl( \alpha C_{max} + (1- \alpha) C_{ \infty}) $ is equal to $ \alpha C_{max} + (1- \alpha)cl(C_{ \infty})$. 
Hence, \eqref{eqn:pf_3_3} implies 
\begin{align} \label{eqn:thm5}
    \alpha C_{max} + (1- \alpha) cl(C_{ \infty}) \subseteq cl(C_{ \infty}) = \alpha cl(C_{ \infty}) + (1- \alpha) cl(C_{ \infty}).
\end{align}
Since $ cl(C_{ \infty})$ is convex and compact, by order cancellation theorem \cite[Theorem 4]{grzybowski2020order}, \eqref{eqn:thm5} implies $ \alpha C_{max} \subseteq \alpha cl(C_{ \infty})$ and thus $C_{max} \subseteq cl(C_{ \infty})$. Also, since $C_{max}$ is the maximal RCIS, $C_{ \infty} \subseteq C_{ max}$. Thus, $C_{max} = cl(C_{ \infty})$.
\end{proof}

\begin{proof}[\textbf{Proof of Theorem \ref{thm:conv_inv_ext}}]
Let $\beta_{0}$ be the infimum $ \beta$ such that $ C_0 \subseteq \beta cl(C_{ \infty})$. We first want to show that $ \beta_0 < 1$. By Lemma \ref{lem:c_inf}, $cl(C_{ \infty})$ is a compact and convex set and there exists $ \epsilon>0$ such that $C_{0}+ B_{ \epsilon}(0) \subseteq cl(C_{ \infty})$. Since $cl(C_{ \infty})$ is bounded, there exists $ \alpha>0$ such that $ \alpha cl(C_{ \infty}) \subseteq B_{ \epsilon}(0)$. Thus, $C_{0} + \alpha cl(C_{ \infty}) \subseteq cl( C_{ \infty}) = \alpha cl(C_{ \infty}) + (1- \alpha) cl(C_{ \infty})$. Since $ \alpha cl(C_{ \infty})$ is a nonempty compact set and $ (1- \alpha) cl(C_{ \infty})$ is convex, by the order cancellation theorem (Theorem 4 in \cite{grzybowski2020order}), $C_0 \subseteq (1- \alpha) cl(C_{ \infty})$. Thus, $ \beta_0 \leq  1- \alpha < 1$.

	Pick $ \beta_1$ in $( \beta_0, 1)$. For now, let us assume that there exists a $N>0$  such that 
	\begin{align} \label{eqn:beta_1_ext} 
	   C_{N}= Pre_{\Sigma}^{N}(C_0,S_{xu},D) \supseteq \beta_1 cl(C_{ \infty}).
	\end{align} 
Then, 
   $ Pre_{\Sigma}^{N}( \beta_0 cl(C_{ \infty}),S_{xu},D) \supseteq \beta_1 cl(C_{ \infty})$.
Let $ \lambda = \beta_0/ \beta_1$. Then, $ \beta_1 cl(C_{ \infty})$ is an $N$-step $\lambda$-contractive set.

Now it is left to show that there exists a $N>0$  such that \eqref{eqn:beta_1_ext} holds. It is easy to show that (i) $\{C_{k}\}_{k=1}^{ \infty} $ is an expanding family of convex compact  sets. (ii) Since ${C_{ \infty} = \cup_{k=1}^{ \infty} C_{k}}$ contains $0$ in the interior, by Lemma \ref{lem:c_inf} (c),  there exists $k_0$ such that $ C_{k_0}$ contains $ 0$ in the interior. (iii) $C_{ \infty}$ is bounded since $S_{xu}$ is bounded.  Finally, since $cl(C_{ \infty})$ contains $0$ in the interior,  $cl(C_{ \infty}) = \beta_1 cl(C_{ \infty}) + (1- \beta_1)cl( C_{ \infty}) \supseteq \beta_1 cl(C_{ \infty}) + B_{ \epsilon}(0)$ for some small $ \epsilon$. Thus, $ \beta_1 cl(C_{ \infty})$ is a compact set in the interior $int(cl(C_{ \infty}))$. In the proof of Lemma \ref{lem:h_conv}, we have shown that $ int(cl(C_{ \infty})) = int(C_{ \infty})$. Thus, (iv) $ \beta_1 cl(C_{ \infty})$ is a compact set in the interior of $C_{ \infty}$. Then, by Lemma \ref{lem:h_conv}, there exists $N$ such that $C_{N} \supseteq \beta_1 cl(C_{ \infty})$.  
\end{proof}

\begin{proof}[\textbf{Proof of Theorem \ref{thm:conv_ext} }]
		Define 
	\begin{align}
	    \lambda_{0,1} =  \frac{ \lambda_0 - \lambda \gamma}{1- \lambda \gamma} ,\ 
		\lambda_{0,2} = \frac{ \lambda \gamma (1- \lambda_0)}{1-  \lambda\gamma}. 
	\end{align}
	It is easy to check that $ \lambda_{0,1}\geq 0$, $ \lambda_{0,2}\geq 0$, $ \lambda_0 = \lambda_{0,1} + \lambda_{0,2}$ and $ 1- \lambda_{0,1} = \lambda_{0,2}/( \lambda \gamma)$.  Thus, by Lemmas \ref{lem:pre_decomp_cvx} and \ref{lem:pre_linear}, we have
	\begin{align*}
	&Pre_{\Sigma}^{N}( \lambda_0 C, S_{xu},D)\\
		\supseteq & Pre_{\Sigma}^{N} ( \lambda_{0,1} C, \lambda_{0,1}S_{xu}, \lambda_{0,1}D ) + Pre_{\Sigma}^{N} ( \lambda_{0,2} C, \frac{ \lambda_{0,2}}{ \lambda \gamma} S_{xu}, \frac{ \lambda_{0,2}}{ \lambda \gamma} D)\\
		=&  \lambda_{0,1}Pre_{\Sigma}^{N} (C, S_{xu}, D ) + \frac{ \lambda_{0,2}}{ \lambda \gamma}  Pre_{\Sigma}^{N} ( \lambda \gamma C, S_{xu}, D)\\
		\supseteq & \lambda_{0,1} C + \frac{\lambda_{0,2}}{ \lambda \gamma} \gamma C = g( \lambda_0) C.
	\end{align*}
The inclusion in last row above holds since $ \gamma C$ is $N$-step $ \lambda$-contractive.
\end{proof}

\balance
\end{document}